\documentclass[a4paper,UKenglish,cleveref, autoref, thm-restate]{lipics-v2021}

\pdfoutput=1 
\hideLIPIcs  

\graphicspath{{./figures/}}

\title{Greedy Monochromatic Island Partitions}

\author{Steven {van den Broek}}{TU Eindhoven, the Netherlands}{s.w.v.d.broek@tue.nl}{https://orcid.org/0009-0005-6677-3916}{}

\author{Wouter Meulemans}{TU Eindhoven, the Netherlands}{w.meulemans@tue.nl}{https://orcid.org/0000-0002-4978-3400}{W. Meulemans is partially supported by the Dutch Research Council (NWO) under project number VI.Vidi.223.137.}

\author{Bettina Speckmann}{TU Eindhoven, the Netherlands}{b.speckmann@tue.nl}{https://orcid.org/0000-0002-8514-7858}{}

\authorrunning{S. van den Broek, W. Meulemans and B. Speckmann}

\Copyright{Steven van den Broek, Wouter Meulemans, and Bettina Speckmann}

\ccsdesc[500]{Theory of computation~Approximation algorithms analysis}
\ccsdesc[500]{Theory of computation~Computational geometry}

\keywords{Approximation algorithms, colored points, partitions}
\category{}
\relatedversion{}
\nolinenumbers

\EventEditors{John Q. Open and Joan R. Access}
\EventNoEds{2}
\EventLongTitle{42nd Conference on Very Important Topics (CVIT 2016)}
\EventShortTitle{CVIT 2016}
\EventAcronym{CVIT}
\EventYear{2016}
\EventDate{December 24--27, 2016}
\EventLocation{Little Whinging, United Kingdom}
\EventLogo{}
\SeriesVolume{42}
\ArticleNo{23}

\usepackage{microtype}

\newcommand{\OptL}{\text{Opt}_{\text{L}}}

\newcommand{\OptDJ}{\text{Opt}_{\text{P}}}
\newcommand{\OptO}{\text{Opt}_{\text{C}}}
\newcommand{\CH}[1]{\mathcal{CH}(#1)}
\newcommand{\N}{\mathbb{N}}
\newcommand{\R}{\mathbb{R}}
\newcommand{\set}[1]{\{#1\}}
\newcommand{\FlatTree}[1]{%
    \ifthenelse{ \equal{#1}{} }
        {\textsf{\textsc{\upshape FlatTree}}}
        {\textsf{\textsc{\upshape FlatTree}}(#1)}
}

\begin{document}

\maketitle
\begin{abstract}
Constructing partitions of colored points is a well-studied problem in discrete and computational geometry. 
We study the problem of creating a minimum-cardinality partition into monochromatic islands. 
Our input is a set $S$ of $n$ points in the plane where each point has one of $k \geq 2$ colors.
A set of points is monochromatic if it contains points of only one color. 
An island $I$ is a subset of $S$ such that $\CH{I} \cap S = I$, where $\CH{I}$ denotes the convex hull of $I$. 
We identify an island with its convex hull; therefore, a partition into islands has the additional requirement that the convex hulls of the islands are pairwise-disjoint. 
We present three greedy algorithms for constructing island partitions and analyze their approximation ratios.
\end{abstract}

\section{Introduction}
Constructing partitions of colored points is a well-studied problem in discrete~\cite{monochromatic-parts, survey-discrete-red-blue-new} and computational geometry~\cite{geometric-partitions, separate-by-lines-FPT, two-disjoint-disks, diverse-partition}.
The colors of the points can be present in the constraints and the optimization criterion in different ways.
For example, one may require the partition to be \emph{balanced}---see the survey by Kano and Urrutia~\cite{survey-discrete-red-blue-new} for many such instances---or \emph{monochromatic}~\cite{geometric-partitions, separate-by-lines-FPT, two-disjoint-disks, monochromatic-parts}.
Alternatively, one may want to minimize or maximize the \emph{diversity}~\cite{diverse-partition} or \emph{discrepancy}~\cite{coarseness, coarseness-new} of the partition.
Furthermore, one can use different geometries to partition the points, such as triangles~\cite{geometric-partitions}, disks~\cite{two-disjoint-disks}, or lines~\cite{separate-by-lines-FPT}.

We study the problem of creating a minimum-cardinality partition into \emph{monochromatic islands}~\cite{optimal-islands}.
Our input is a set $S$ of $n$ points in the plane where each point has one of $k \geq 2$ colors.
A set of points is \emph{monochromatic} if it contains points of only one color. 
An \emph{island}~$I$ is a subset of $S$ such that $\CH{I} \cap S = I$, where $\CH{I}$ denotes the convex hull of $I$.
We identify an island with its convex hull; therefore, a partition into islands has the additional requirement that the convex hulls of the islands are pairwise-disjoint.

\subparagraph{Related work.}
Bautista-Santiago et al.~\cite{optimal-islands} study islands and describe an algorithm that can find a monochromatic island of maximum cardinality in $O(n^3)$ time, improving upon an earlier $O(n^3 \log n)$ algorithm~\cite{DBLP:journals/comgeo/Fischer97}.
Dumitrescu and Pach~\cite{monochromatic-parts} consider monochromatic island partitions and prove how many islands are sufficient and sometimes necessary for different types of input.
Bereg et al.~\cite{coarseness} use island partitions to define a notion of \emph{coarseness} that captures how blended a set of red and blue points are. 
Agarwal and Suri~\cite{geometric-partitions} study the following problem: given red and blue points, cover the blue points with the minimum number of pairwise-disjoint monochromatic triangles.
They prove that this problem is NP-hard and describe approximation algorithms.
Their NP-hardness reduction can be used to prove that covering and partitioning points of only one color into the minimum number of monochromatic islands---the points of the other colors serving only as obstacles---is NP-hard, as observed by Bautista-Santiago et al.~\cite{optimal-islands}.
We suspect that the problem we study is NP-hard as well, which motivates us to focus on approximation algorithms.

\subparagraph{Overview.}
In the remainder, we consider only monochromatic islands. We denote by $\OptDJ$ the minimum cardinality of an island partition of $S$.
In the following sections, we use three greedy algorithms---\emph{disjoint-greedy}, \emph{overlap-greedy}, and \emph{line-greedy}---to construct island partitions.
Disjoint-greedy creates an island partition by iteratively picking the island that covers most uncovered points and does not intersect any island chosen before.
We prove that disjoint-greedy has an approximation ratio of $\Omega(n/\log^2{n})$.
The overlap-greedy algorithm greedily constructs an $O(\log{n})$-approximation of the minimum-cardinality island cover.
We prove that any algorithm that transforms an island cover returned by overlap-greedy into an island partition has approximation ratio $\Omega(\sqrt{n})$, and describe one such algorithm that has approximation ratio $O(\OptDJ^2 \log^2{n})$.
Lastly, we investigate the relation between constructing a minimum-cardinality island partition and finding the minimum number of lines that separate the points into monochromatic regions.
In particular, we show that greedily choosing the line that separates most pairs of points of different color induces an $O(\OptDJ \log^2{n})$-approximation to the minimum-cardinality island partition.
Figure~\ref{fig:algorithms} illustrates the greedy algorithms.

\begin{figure}[tb]
    \centering
    \includegraphics[page=2]{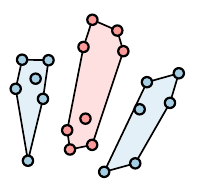}
    \includegraphics[page=3]{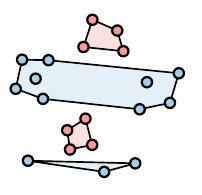}
    \includegraphics[page=4]{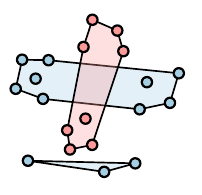}
    \includegraphics[page=6]{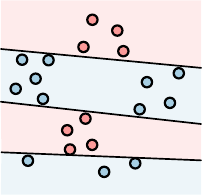}
    \caption{Left: optimal island partition; middle-left: disjoint-greedy island partition; middle-right: overlap-greedy island cover; right: line-greedy separating lines.}
    \label{fig:algorithms}
\end{figure}

\section{Disjoint-Greedy}\label{sec:disjoint-greedy}
We first sketch our lower bound construction that shows that disjoint-greedy has an approximation ratio of $\Omega(n/\log^2{n})$.
Consider a family of problem instances that have the form of two opposing complete binary trees of height~$\ell$ (Figure~\ref{fig:disjoint-greedy:instance}). 
Sets of points are placed close together at the nodes of these trees.
The idea is that by placing sufficiently many points at the nodes, and by placing obstacle points appropriately, the disjoint-greedy algorithm iteratively picks points of two opposing nodes such that the problem instance is split into two symmetric nearly independent parts that have nearly the same structure as the original instance.
This results in disjoint-greedy returning a partition into $\Omega(2^\ell)$ islands (Figure~\ref{fig:disjoint-greedy:lower:greedy}).
However, there exists a partition such that each layer in the tree consists of a constant number of islands, resulting in $O(\ell)$ islands in total (Figure~\ref{fig:disjoint-greedy:lower:optimal}).
In our construction, the number of red points at a node at height $i \in \set{0, \dots, \ell - 1}$ is $2^{i+6}$ and the number of blue points at a node is constant.
Therefore, the number of points in a layer is $\Theta(2^\ell)$.
As there are $\Theta(\ell)$ layers, there are $\Theta(\ell \cdot 2^\ell)$ points in total and the approximation ratio of disjoint-greedy is $\Omega(2^\ell / \ell) = \Omega(n / \log^2{n})$. 

\begin{figure}[tb]
    \centering
    \includegraphics[page=1]{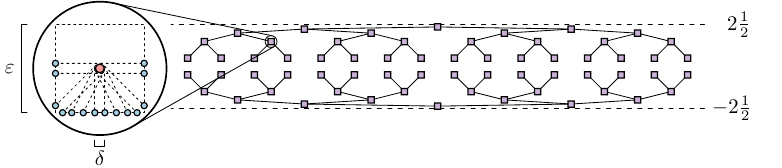}
    \caption{The problem instance for $\ell = 5$. The lines in the figure are not part of the problem instance, but illustrate its structure. The purple squares represent red and blue points lying close together inside a square. The red disk inside the square represents many red points placed together inside a disk. The centers of the purple squares lie within the strip bounded by the two dashed lines.}
    \label{fig:disjoint-greedy:instance}
\end{figure}
\begin{figure}[tb]
    \centering
    \includegraphics[page=2]{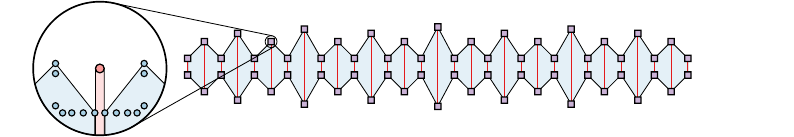}
    \caption{The solution returned by disjoint-greedy.}
    \label{fig:disjoint-greedy:lower:greedy}
\end{figure}
\begin{figure}[t!]
    \centering
    \includegraphics[page=3]{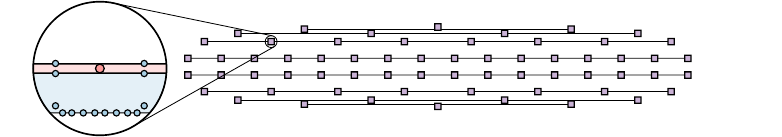}
    \caption{An alternative solution, serving as an upper bound for the optimal solution.}
    \label{fig:disjoint-greedy:lower:optimal}
\end{figure}

We provide a formal definition of the family of problem instances we have described.
There are three parameters: $\ell \in \N_{\geq 1}$, $\varepsilon \in \R_{> 0}$ and $0 < \delta < \varepsilon$. 
We denote a particular problem instance with $\FlatTree{\ell, \varepsilon, \delta}$, which contains the points we describe next.
Define
\[
    P^+_{i, j} =
    (2^i + j \cdot 2^{i+1}, 2\frac{1}{2} - 2^{1-i}) 
    ,
    \quad \text{for } 
    i \in \set{0, \dots, \ell - 1}, j \in \set{ 0, \dots, 2^{\ell - 1 - i} - 1}.
\]
Referring to Figure~\ref{fig:disjoint-greedy:instance}, $P_{i, j}^+$ is the position of the $j$'th purple square at height~$i$ in the top tree.
Define $r_i = 2^{i+6}$, which will determine the number of red points in a purple square at height~$i$. 
Let $B_r(x)$ denote the open ball centered at $x$ with radius $r$.
The red points corresponding to $P_{i, j}^+$ are
\begin{align*}
    R^+_{i, j} &= 
    \ 
    \text{a set of } r_j \text{ distinct points in } B_{\delta/2}(P_{i, j}^+)
    .
\end{align*}
Next, we describe the blue points $B_{i, j}^+$ corresponding to $P_{i, j}^+$.
Let $\overline{Q}_w(x)$ denote the closed square centered at $x$ of width $w$.
The outer common tangents of $B_{\delta / 2}(P_{i, j}^+)$ and balls belonging to specific other purple squares are intersected with the boundary of $\overline{Q}_{\varepsilon}(P_{i, j}^+)$ to give rise to the blue points corresponding to $P_{i, j}^+$ (Figure~\ref{fig:disjoint-greedy:instance}).
Let a set with a $-$ superscript denote taking the corresponding set with a $+$ superscript but mirroring its points over the $y$-axis.
Then, outer common tangents are created with balls of the following squares, if they exist:
\[
P_{i, j-1}^+,\, P_{i, j+1}^+,\, P_{i, j}^-,\, P_{i-1, 2j}^\times,\, P_{i-1, 2j+1}^\times,
\]
where $\times \in \set{-, +}$.
$\FlatTree{\ell, \varepsilon, \delta}$ consists of red points $R_{i, j}^+$ and $R_{i, j}^-$ and blue points $B_{i, j}^+$ and $B_{i, j}^-$ for all $i, j$.

Intuitively, we want to prove that there exists a $\FlatTree{}$ instance for which disjoint-greedy returns a partition that contains at least the vertical red islands illustrated in Figure~\ref{fig:disjoint-greedy:lower:greedy}. 
We define these vertical red islands formally as follows.
\begin{align*}
    V_{i, j} &= R^-_{i, j} \cup R^+_{i, j}\\
    V_i &= \set{V_{i, j} \mid j \in \set{0,\, \dots,\, 2^{\ell - 1 - i} - 1}}\\
    V &= \set{V_{i, j} \mid i \in \set{0, \dots, \ell - 1}, j \in \set{0,\, \dots,\, 2^{\ell - 1 - i} - 1}}
\end{align*}
We can now state the lemma.
\begin{lemma}\label{lemma:disjoint-greedy:greedy-sol-lower}
For any $\ell \in \N_{\geq 1}$, there exist an $\varepsilon > 0$ and a $0 < \delta < \varepsilon$ such that for input $\FlatTree{\ell, \varepsilon, \delta}$
disjoint-greedy returns a partition $\mathcal{P}$ that is a superset of $\,V$. 
\end{lemma}

\begin{proof} Let $\varepsilon$ and $\delta$ be such that the constraints stated in the proof are satisfied.
We prove that the following loop invariant holds during the execution of the disjoint-greedy algorithm.

\begin{description}
    
\item[Invariant.] For $k \in \set { 0, \dots, \ell }$ we define invariant $\mathcal{I}(k)$ as: before the start of iteration $2^{k}$, the current island partition $\mathcal{P}$ satisfies:
\[
\mathcal{P} = \bigcup_{j = \ell - k}^{\ell - 1} \ V_j.
\]

\item[Initialization.] 
At the start of iteration $2^0 = 1$ we have $\mathcal{P} = \emptyset$, so $\mathcal{I}(0)$ holds.

\item[Maintenance.] Assume $\mathcal{I}(k)$ holds for some $k \in \set{0, \dots, \ell - 1}$. 
We prove that $\mathcal{I}(k+1)$ holds. 
For ease of notation we define $\ell' = \ell - k - 1$.
Let $\mathcal{C}$ denote the set of candidate islands in this iteration---that is, the set of islands that do not intersect any island chosen before.
We claim that the $2^k$ largest islands of $\mathcal{C}$ are the set $V_{\ell'}$. Clearly, the islands in $V_{\ell'}$ do not intersect. Therefore, in the next $2^k$ iterations, disjoint-greedy iteratively chooses the islands $V_{\ell'}$, and $\mathcal{I}(k+1)$ holds.

\item[Termination.] We have proven, in particular, that $\mathcal{I}(\ell)$ holds. This implies that the partition $\mathcal{P}$ that disjoint-greedy returns contains the $2^{\ell - 1}$ islands of set $V$.

\item[Proof of the claim.] 
First note that each island in the set $V_{\ell'}$ has cardinality $2r_{\ell'} = 2^{\ell - k + 6}$. The islands that are part of $\mathcal{P}$ at the start of iteration $k$ induce $2^k$ non-disjoint regions as shown in Figure~\ref{fig:disjoint-greedy:regions}. We distinguish between inter- and intra-region islands. We arrange the islands $\mathcal{C} \setminus V_{\ell'}$ into three categories and prove each type of island has cardinality less than $2^{\ell - k + 6}$.

\begin{figure}[b]
    \centering
    \includegraphics[page=1]{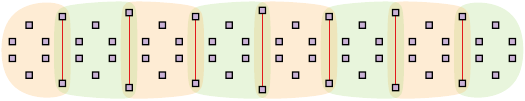}
    \caption{Regions for $k=3$}
    \label{fig:disjoint-greedy:regions}
\end{figure}

\begin{description}
\item[Blue intra-region island.]
Each region contains at most $2^{\ell' + 2} + 2$ purple squares.
Because each square contains at most fourteen blue points, there are at most $2^{\ell - k + 5} + 32$ blue points in any region, which is less than $2^{\ell - k + 6}$ for any $k \in \set{0, \dots, \ell - 1}$.

\item[Red island.]
Any red island that uses points from at most two purple squares contains strictly less points than those in $V_{\ell'}$. 
Any red island that uses points from more than two purple squares is significantly constrained by blue points.
In particular, observe that 
there exist $\varepsilon > 0$ and $0 < \delta < \varepsilon$ such that the following statements hold.
\begin{enumerate}
\item 
\label{observation:1}
Any set that contains red points from distinct $R^{\times}_{i, j_1}$, $R^{\times}_{i, j_2}$ and $R^{*}_{i', j_3}$, with $i \neq i'$ and ${\times}, {*} \in \{{+}, {-}\}$ is not an island.
\item 
\label{observation:2}
Any set that contains red points from distinct $R^{-}_{i, j}$, $R^{+}_{i, j}$ and $R^{\times}_{i', j'}$, with ${\times} \in \set{{-}, {+}}$ is not an island.
\item
\label{observation:3}
If $\ell' > 0$ then any set that contains red points from distinct $R^{\times_1}_{\ell', j_1}$, $R^{\times_2}_{\ell' - 1, j_2}$, $R^{\times_3}_{i, j_3}$ that are part of the same region, with ${\times}_r \in \set{{-}, {+}}$ and $j_2 \in \set{2 j_1, 2 j_1 + 1}$, is not an island.
\end{enumerate}
Consider an arbitrary red island in $\mathcal{C} \setminus V_{\ell'}$. If it contains red points from distinct $R^{\times}_{i, j_1}$ and $R^{\times}_{i, j_2}$, so from two purple squares of the same \textit{layer}, then by Observation~\ref{observation:1} it contains points only from purple squares of the same layer. 
This would imply the cardinality of the island is at most $r_{\ell'} < 2r_{\ell'}$.
Thus, assume the red island contains points from only one purple square per layer.
By Observation~\ref{observation:2}, if the red island contains points from $R_{i, j}^-$ and $R_{i, j}^+$, so from two \textit{opposing} purple squares, it cannot contain points from any other purple square.
This would result in a cardinality of at most $2r_{\ell' - 1} = r_{\ell'} < 2r_{\ell'}$, or zero if $\ell' = 0$.
Thus, suppose the red island does not contain points of opposing purple squares.
If the red island does not contain points from $R^{\times}_{\ell', j}$ for some $\times \in \set{-, +}$ and $j$, then it has cardinality at most
\[
2\sum_{i = 0}^{\ell' - 1} r_i \leq 2r_{\ell'} - 2 < 2r_{\ell'},
\]
so assume the red island contains points from one purple square $R^{\times}_{\ell', j}$.
If the red island contains points from a purple square $R^{*}_{\ell' - 1, j'}$, then by Observation~\ref{observation:3}, it cannot contain points from another purple square. 
This would result in an island of cardinality at most $r_{\ell'} + r_{\ell' - 1} < 2r_{\ell'}$.
Therefore, assume the red island does not contain points from a purple square $R^{*}_{\ell' - 1, j'}$.
Then, the red island has cardinality at most
\[
r_{\ell'} + 2\sum_{i = 0}^{\ell' - 2} r_i \leq r_{\ell'} + 2r_{\ell' - 1} - 2 = 2r_{\ell'} - 2 < 2r_{\ell'}.
\]

\item[Blue inter-region island.]
We claim that for sufficiently small $\varepsilon$, any blue inter-region island uses points from at most three purple squares. Thus, any such island has size at most 42, which is smaller than $2^{\ell - k + 6}$ for any $k \in \set{0, \dots, \ell - 1}$.
Next, we prove the claim.
Without loss of generality, assume that the blue inter-region island island uses points only from squares that all lie above the $x$-axis.
For squares that all lie below the $x$-axis the argument is symmetric. Blue inter-region islands that cross the $x$-axis clearly cannot contain more points than islands that do not.
To see why any blue inter-region island uses points from at most three purple squares, consider the ray starting at a blue point in a purple square at some height $j \in \set{0, \dots, \ell - 1}$ and going through a blue point in a second purple square at some greater height that is closest in terms of $x$-coordinate. The distance between the two points along the axes is $\Delta x \leq 2^{j} + \varepsilon$ and $\Delta y \geq 2^{-j} - \varepsilon$. Constrain $\varepsilon < 2^{-j-3}$, then $\Delta x < 2^j + 2^{-j-3}$ and $\Delta y > 2^{-j} - 2^{-j-3}$.
At least $a := 2^{j+1} -\varepsilon > 2^{j+1} - 2^{-j-3}$ units in the horizontal direction after the second purple square, the first purple square appears that is at a greater height than the second square the ray visited, if one exists in that direction. At that point, the ray is at $y$-coordinate at least
\[
y_\text{start} + \Delta y + \frac{\Delta y}{\Delta x} \cdot a > 2\frac{1}{2} - 2^{1-j} + 2^{-j} - 2^{-j-3} + \frac{2^{-j} - 2^{-j-3}}{2^j + 2^{-j-3}} \cdot (2^{j+1} - 2^{-j-3}).
\]
Above, $y_\text{start}$ is the $y$-coordinate of the square that the ray starts at. We verify that the ray exceeds $y$-coordinate $2\frac12$ at this point, implying that no purple squares lie above the ray. Rewriting the inequality
\[
y_\text{start} + \Delta y + \frac{\Delta y}{\Delta x} \cdot a > 2\frac12
\]
yields, after a sequence of standard arithmetic operations, the following inequality:
%
\[
2^{-2j-2} < \frac{5}{8}.
\]
This inequality holds for all $j \geq 0$.
This implies that any island that contains a blue point from a square at height at most $\ell'$ is an intra-region island (Figure~\ref{fig:disjoint-greedy:inter-intra} left).
Therefore, consider an island that contains blue points from squares at height greater than $\ell'$, if one exists.
Consider the square at the lowest height from which the island contains points.
Then it is clear that only points of the at most two purple squares of greater height that are closest to it may be part of the island (Figure~\ref{fig:disjoint-greedy:inter-intra} right).
Hence, indeed, any blue inter-region island has size at most~$42$.

\begin{figure}
    \centering
    \includegraphics[page=2]{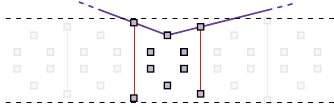}
    \hspace{5mm}
    \includegraphics[page=3]{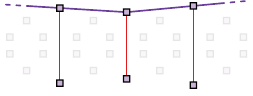}
    \caption{Figures illustrating that any blue inter-region island uses points from at most three purple squares. Left: intra-region island; right: inter-region island.}
    \label{fig:disjoint-greedy:inter-intra}
\end{figure}

\end{description}

\end{description}
\end{proof}

\begin{lemma}\label{lemma:disjoint-greedy:lower}
    The disjoint-greedy algorithm has approximation ratio $\Omega(\frac{n}{\log^2{n}})$.
\end{lemma}
\begin{proof}
Let $\ell$ be arbitrary and let $\varepsilon$ and $\delta$ be such that Lemma~\ref{lemma:disjoint-greedy:greedy-sol-lower} holds. Consider input instance $\FlatTree{\ell, \varepsilon, \delta}$. Lemma~\ref{lemma:disjoint-greedy:greedy-sol-lower} implies that disjoint-greedy returns a partition with $\Omega(2^\ell)$ islands. A solution as shown in Figure~\ref{fig:disjoint-greedy:lower:optimal}, which uses at most three islands per layer, has cardinality $O(\ell)$. Thus, the approximation ratio of disjoint-greedy is $\Omega(2^\ell/\ell)$. 
The total number of points $n$ is $\Theta(\ell \cdot 2^\ell)$ as the height of the trees is $\ell$ and there are $\Theta(2^\ell)$ points at the same height.
Writing $2^\ell$ and $\ell$ as functions of $n$, it is readily seen that $2^\ell = \Omega(n/\log{n})$ and $\ell = O(\log{n})$. 
Thus, the approximation ratio of disjoint-greedy is $\Omega(n/\log^2{n})$.
\end{proof}

\section{Overlap-Greedy}\label{sec:overlap-greedy}
A greedy algorithm that iteratively picks the island that covers most uncovered points results in an $O(\log{n})$ approximation to the minimum-cardinality island \emph{cover}.
This follows immediately from viewing the problem as a set cover problem, where islands form the sets.
We refer to this greedy algorithm as \textit{overlap-greedy}. 
Below, we explore how to transform the island cover returned by overlap-greedy into an island partition.
We assume that the greedy algorithm breaks ties by choosing an island that covers the fewest previously covered points.

We first define the relation between island covers and partitions based on them. 
Intuitively, islands that intersect can be transformed into a set of pairwise-disjoint islands by splitting them. 
In the transformation, each island has a corresponding family of islands into which it is split. 
The union of such a family should be a subset of the original island---a subset, not equal, because a point that was originally covered by multiple islands should be part of exactly one island after the transformation.
This motivates the following definition.
\begin{definition}[Compatible]
Families $\mathcal{I}'_1, \dots, \mathcal{I}'_{m}$ are \textit{compatible} with islands $I_1, \dots, I_m$ if:
\begin{itemize}
    \item Families $\mathcal{I}'_1, \dots, \mathcal{I}'_{m}$ cover the same points as $I_1, \dots, I_m$: $\bigcup_k \bigcup \mathcal{I}'_k = \bigcup_i I_i$;
    \item For every $i$, we have $\bigcup\mathcal{I}'_i \subseteq I_i$;
    \item Islands $\bigcup_k \mathcal{I}'_k$ are pairwise-disjoint.
\end{itemize}
Islands $I'_1, \dots, I'_{m'}$ are compatible with islands $I_1, \dots, I_m$ if there exists a partition of $\set{I'_1, \dots, I'_{m'}}$ into families that are compatible with $I_1, \dots, I_m$.
\end{definition}
Thus, we arrive at the following problem: given $m$ islands $I_1, \dots, I_m$ obtained by overlap-greedy, find compatible families $\mathcal{I}'_1, \dots, \mathcal{I}'_{m}$ with $|\bigcup_k \mathcal{I}_k'|$ minimum.
Solving this problem optimally is non-trivial. A natural approach to tackle the problem is to create an arrangement of the islands $I_1, \dots, I_m$ and extract compatible families from that arrangement. However, two minor issues arise: the faces in the arrangement may not be convex, and the quality of the solution is not immediately clear as the number of faces in the arrangement may be arbitrarily greater than $m$.
To resolve these issues, we build an arrangement iteratively. Before describing this process, we give a lower bound on the approximation ratio of algorithms that return islands compatible with those returned by overlap-greedy.

\begin{lemma}
Any algorithm that returns islands compatible with those returned by overlap-greedy has approximation ratio $\Omega(\max\{\OptDJ, \sqrt{n}\})$.
\end{lemma}
\begin{proof}
Let $k \in \N_{\geq 1}$. Consider a problem instance that consists of $k$ evenly spaced vertical blue lines each formed by $k + 1$ evenly spaced blue points, and symmetrically $k$ evenly spaced horizontal red lines (Figure~\ref{fig:overlap-greedy-compatible-lower}). The cover returned by overlap-greedy has cardinality $2k$ and is induced by exactly those lines that were just described (Figure~\ref{fig:overlap-greedy-compatible-lower}, left). Any partition that is compatible with the overlap-greedy cover has cardinality at least $2k + k^2$ (Figure~\ref{fig:overlap-greedy-compatible-lower}, middle). Indeed, each intersection between two lines in the overlap-greedy cover forces an additional island in a compatible partition. An optimal island partition has cardinality $\OptDJ = 2k+1$ and consists of either all horizontal or all vertical islands (Figure~\ref{fig:overlap-greedy-compatible-lower}, right).
The number of points $n = O(k^2)$. 
Thus, the approximation ratio of an algorithm that produces solutions compatible with that of overlap-greedy is at least $\frac{2k+k^2}{2k+1} = \Omega(k) = \Omega(\OptDJ) = \Omega(\sqrt{n}) = \Omega(\max\{\OptDJ, \sqrt{n}\})$.
\end{proof}

\begin{figure}[tb]
    \hspace*{\fill}
    \includegraphics[page=2]{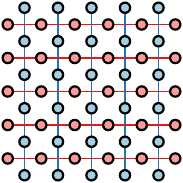}
    \hfill
    \includegraphics[page=3]{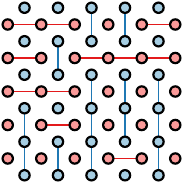}
    \hfill
    \includegraphics[page=4]{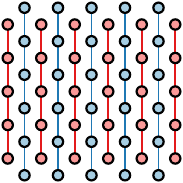}
    \hspace*{\fill}
    \caption{
        A lower bound on the cardinality of solutions compatible with islands returned by overlap-greedy.
        Left: overlap-greedy cover; middle: a partition compatible with the overlap-greedy cover; right: an optimal island partition.
    }
    \label{fig:overlap-greedy-compatible-lower}
\end{figure}

\subsection{Upper Bound}
In this section, we describe an algorithm for transforming an island cover of cardinality $q$ returned by overlap-greedy into a compatible island partition of cardinality $O(q^2 \OptDJ)$. 
The proof of this bound requires a lower bound on the cardinality $\OptDJ$ of the optimal partition, which we state first.

Let $P$ be a set of points in convex position and let $p_1, \dots, p_n$ be a cyclic ordering of $P$ agreeing with $\delta(\CH{P})$, the boundary of the convex hull of $P$. 
Then we say that $P$ has alternating colors if for any $i$ the color of $p_i$ is distinct from the color of its predecessor and successor in the cyclic order.
The following holds.

\begin{lemma}\label{lemma:kgon}
    If an instance contains $2k$ points in convex position with alternating colors, then $\OptDJ \geq k + 1$, where $\OptDJ$ denotes the cardinality of the optimal island partition of the instance.
\end{lemma}
    This lemma is proven by Dumitrescu and Pach~\cite{monochromatic-parts} (Theorem~1). Though their theorem concerns points on a circle, the proof holds more generally for points lying in convex position.

As mentioned earlier, our algorithm for creating a compatible island partition from an island cover works in an iterative manner. Throughout the iterations, we keep track of a restricted planar subdivision, which we call an \emph{island arrangement}, to bound the cardinality of the island partition constructed by the algorithm. To simplify the arguments, we assume all islands in the island cover have cardinality at least three and that no three points are collinear. Then, a compatible island partition can be created from the faces of the island arrangement.
See Figure~\ref{fig:bold-example} for an overview of the transformation from island cover to a compatible island partition.

We now define the notion of an \emph{island arrangement}.
In the following, vertices, edges, and faces of a planar subdivision are collectively referred to as \emph{features}.

\begin{figure}[t]
    \centering
    \hspace*{\fill}
    \includegraphics[page=1]{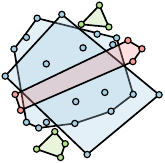}
    \hfill
    \includegraphics[page=2]{bold-example.pdf}
    \hfill
    \includegraphics[page=3]{bold-example.pdf}
    \hspace*{\fill}
    \caption{Left: an island cover $I_1, \dots I_m$ returned by overlap-greedy; middle: an island arrangement of $I_1, \dots, I_m$; right: an island partition compatible with $I_1, \dots, I_m$ induced by the arrangement. 
    }
    \label{fig:bold-example}
\end{figure}


\begin{definition}[Island arrangement]\label{def:island-arrangement}
An island arrangement of islands $I_1, \dots, I_i$ is a planar subdivision with the following additional requirements:
\begin{itemize}
    \item Bounded faces are convex;
    \item Every bounded feature is a subset of $\CH{I_j}$ for some $1 \leq j \leq i$;
    \item For every $1 \leq j \leq i$, $\CH{I_j}$ is covered by bounded features.
\end{itemize}
\end{definition}

Let $I_1, \dots, I_m$ be the islands chosen by overlap-greedy for some set of points $S$. Let $U_i = I_i \setminus \bigcup_{j < i} I_j$ be the set of uncovered points island $I_i$ covers. Because islands are chosen greedily by overlap-greedy, they satisfy $|U_i| \geq |U_{i+1}|$ for $i \in \set{1, \dots, m-1}$ and for all $i$ island $I_i$ is such that $|U_i|$ is maximum. By using these properties, we prove the following lemma.
\begin{restatable}{lemma}{giib}\label{lemma:greedy-island-intersection-bound}
Let $\delta$ denote the boundary operator on sets. For distinct $1 \leq i < j \leq m$, the number of intersections between $\delta(\CH{I_i})$ and $\delta(\CH{I_j})$ is at most $2\OptDJ$.
\end{restatable}
\begin{proof}

Without loss of generality, assume points have at most two colors: red and blue, and assume $I_j$ is blue. Furthermore, because the lemma is trivial if $\CH{I_i}$ and $\CH{I_j}$ are disjoint, assume they intersect. We perform a case distinction on the color of $I_i$ to bound the number of intersections $x$ between $\delta(\CH{I_i})$ and $\delta(\CH{I_j})$.

\begin{figure}[tb]
    \begin{subfigure}[t]{0.2\textwidth}
        \includegraphics[page=1]{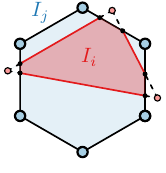}
        \caption{Case $I_i$ is red}
        \label{fig:overlap-greedy-upper:red1}
    \end{subfigure}
    \hfill
    \begin{subfigure}[t]{0.24\textwidth}
        \includegraphics[page=3]{overlap-greedy-upper-new.pdf}
        \caption{Case $I_i$ is red: points in convex position}
        \label{fig:overlap-greedy-upper:red2}
    \end{subfigure}
    \hfill
    \begin{subfigure}[t]{0.2\textwidth}
        \includegraphics[page=5]{overlap-greedy-upper-new.pdf}
        \caption{Case $I_i$ is blue}
        \label{fig:overlap-greedy-upper:blue1}
    \end{subfigure}
    \hfill
    \begin{subfigure}[t]{0.24\textwidth}
        \includegraphics[page=4]{overlap-greedy-upper-new.pdf}
        \caption{Case $I_i$ is blue: points in convex position}
        \label{fig:overlap-greedy-upper:blue2}
    \end{subfigure}

    \caption{Ways in which two islands $I_i$ and $I_j$ chosen by overlap-greedy can intersect.}
    \label{fig:overlap-greedy-upper}
\end{figure}

\begin{description}
\item[Island $I_i$ is red.] 
An example is depicted in Figure~\ref{fig:overlap-greedy-upper:red1}. 
Observe that no points of $I_i$ can be in $\CH{I_j}$ and no points of $I_j$ can be in $\CH{I_i}$. 
This implies that only parts of edges of $\delta(\CH{I_i})$ can intersect $\CH{I_j}$ as in Figure~\ref{fig:overlap-greedy-upper:red1}.
Each part $C_\ell$ is a connected component of $\delta(\CH{I_i}) \cap \CH{I_j}$.
We now argue that $x \leq 2\OptDJ$.
Consider the following set of points.
For every segment $C_\ell$, pick a blue point of $I_j$ on $\delta(\CH{I_j})$ that lies in the strip extending from $C_\ell$ outward from $\CH{I_i}$. 
Order the blue points $b_1, b_2, \dots$, according to their position on $\delta(\CH{I_j})$. 
For every two consecutive $b_i, b_{i+1}$, pick an extreme red point of $I_i$ that lies in the direction perpendicular to $b_i b_{i+1}$ outward from $\CH{I_j}$. 
The blue and red points form $x$ points in convex position with alternating colors (Figure~\ref{fig:overlap-greedy-upper:red2}), so from Lemma~\ref{lemma:kgon} follows that $x \leq 2\OptDJ$.

    

\item[Island $I_i$ is blue.]
If $\CH{I_i}$ cuts only one side of $\CH{I_j}$, the argument is trivial, because $x = 2$ and $\OptDJ \geq 1$, so indeed $x \leq 2\OptDJ$. Therefore, assume that $I_i$ cuts multiple sides of $I_j$ as shown in Figure~\ref{fig:overlap-greedy-upper:blue1}. Each connected component $C_\ell$ of $\delta(\CH{I_i}) \cap \CH{I_j}$ is a chain of vertices $v_1^\ell, \dots, v_m^\ell$ with $m \geq 2$. 
We bound $x$ by applying Lemma~\ref{lemma:kgon} on appropriate points.
Because $I_i$ is maximal, for any point $p \in I_j \setminus I_i$, set $I_i \cup \set{p}$ is not an island. Thus, there must be a red point in $\CH{I_i} \cup \set{p}$ for each $p$. Such a red point lies in one of at most two triangles, as illustrated in Figure~\ref{fig:overlap-greedy-upper:blue1}. For each connected component $C_\ell$, pick one such red point in the half-plane extending from the line through $v_1^\ell v_m^\ell$ outward from $\CH{I_i}$. Order the red points based on their cyclic order on $\delta(\CH{I_i})$ after projecting them. For every two consecutive red points $r_i, r_{i+1}$, pick an extreme point of $I_i$ that lies in the direction perpendicular to $r_i r_{i+1}$ outward from $\CH{I_i}$. We have picked $x$ points that lie in convex position with alternating color (Figure~\ref{fig:overlap-greedy-upper:blue2}). Thus, Lemma~\ref{lemma:kgon} implies that $x \leq 2\OptDJ$.

\end{description}
\end{proof}

Next we show how an island arrangement of $I_1, \dots, I_{i-1}$ with $1 \leq i \leq m$ can be modified into an island arrangement of $I_1, \dots, I_i$ such that the increase in the number of faces is bounded in terms of $i$ and $\OptDJ$. We call this modification an \textit{augmentation} of the arrangement. The following lemma makes one face for the new island $I_i$ and modifies any existing features to make room. We refer to this as a \textit{bold augmentation} of the arrangement. 

\begin{restatable}[Bold augmentation]{lemma}{boldAugmentation}\label{lemma:bold:augmentation}
    Given an island arrangement $A$ of $I_1, \dots, I_{i - 1}$ with $f$ faces, there exists an island arrangement $A'$ of $I_1, \dots, I_{i}$ with at most $f + 2\OptDJ \cdot (i - 1)$ faces such that there is exactly one face whose closure equals $\CH{I_{i}}$.
\end{restatable}
\begin{proof}
\begin{figure}[tb]
    \hspace*{\fill}
    \includegraphics[page=1]{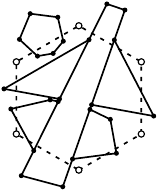}
    \hfill
    \includegraphics[page=2]{bold-augmentation.pdf}
    \hfill
    \includegraphics[page=3]{bold-augmentation.pdf}
    \hspace*{\fill}
    \caption{Examples illustrating Lemma~\ref{lemma:bold:augmentation}. Left: arrangement $A$; island $I_i$ is drawn with a dashed stroke and with white vertices. Middle: arrangement $B$. Right: arrangement $A'$.}
    \label{fig:bold-augmentation}
\end{figure}
If $I_i$ is disjoint from each island $I_1, \dots, I_{i-1}$ then the claim is trivial, so assume it is not.
Let $B$ be the planar subdivision obtained by adding $\delta(\CH{I_i})$ to $A$ (Figure~\ref{fig:bold-augmentation}, middle).
The planar graph corresponding to $B$ has additional vertices, edges and faces induced by $\delta(\CH{I_i})$.
In particular, $B$ has an additional $a + x$ vertices where $a = |I_i \cap \delta(\CH{I_i})|$ and $x$ is the total number of intersections between $\delta(\CH{I_i})$ and $\delta(\CH{I_1}), \dots, \delta(\CH{I_{i-1}})$ (without loss of generality, we assume that at most two boundaries of $I_1, \dots, I_i$ share a common point). 
Furthermore, because $\delta(\CH{I_i})$ consists of $a$ line segments and every intersection with an existing edge creates two additional edges, $B$ has $a + 2x$ edges more than $A$. 
Therefore, Euler's formula implies that $B$ has $x$ faces more than $A$, which is at most $2\OptDJ \cdot (i-1)$ by Lemma~\ref{lemma:greedy-island-intersection-bound}.

Note that $B$ is not an island arrangement of $I_1, \dots, I_i$ according to Definition~\ref{def:island-arrangement} because faces within $\CH{I_i}$ may be non-convex.
Faces of $B$ that are not subsets of $\CH{I_i}$ are all convex.
Indeed, observe that $A$ is an island arrangement, so all of its faces are convex. 
Furthermore, any bounded face of $A$ that is not present in $B$ is split by $\delta(\CH{I_i})$ into convex faces. 
Indeed, suppose a non-convex face was created. 
Then any reflex vertex of such a face, originating from $I_i$, may be removed from $I_i$ to obtain an island covering as many uncovered points as $I_i$.
This contradicts the fact that $I_i$ is minimal (by definition of the overlap-greedy algorithm).
Thus, only faces of $B$ that are subsets of $\CH{I_i}$ may be non-convex.
Define $A'$ as a copy of $B$ where the interior $\CH{I_i}$ is modified to be a single face (Figure~\ref{fig:bold-augmentation}, right).
Then, by our previous argument $A'$ has only convex faces and is an arrangement of $I_1, \dots, I_i$.
Clearly, the number of faces in $A'$ is bounded by the number of faces in $B$, yielding the desired result.
\end{proof}

By repeatedly applying Lemma~\ref{lemma:bold:augmentation}, an island cover returned by overlap-greedy can be transformed into a compatible island partition.
Let \emph{bold overlap-greedy} be the algorithm that first runs overlap-greedy to obtain islands $I_1, \dots, I_m$, then repeatedly applies the bold augmentation step to create an island arrangement $A$ of $I_1, \dots, I_m$, and finally extracts an island partition from the faces of $A$. 
The following bound holds on its approximation ratio.

\begin{restatable}{corollary}{boldOverlapGreedy}
Bold overlap-greedy has approximation ratio $O(\OptDJ^2 \log^2{n})$.
\end{restatable}
\begin{proof}
Overlap-greedy returns $m = O(\OptO \log{n})$ islands with $\OptO$ denoting the cardinality of the optimal island cover. Any island partition is an island cover, so $\OptO \leq \OptDJ$. Lemma~\ref{lemma:bold:augmentation} implies that the partition returned by bold overlap-greedy uses $O(m^2\OptDJ)$ islands. Applying the bound for $m$ yields $O(m^2 \OptDJ) = O(\OptDJ^3 \log^2{n})$. Thus, the approximation ratio of bold overlap-greedy is $O(\OptDJ^2 \log^2{n})$.
\end{proof}

\section{Line-Greedy}\label{sec:line-greedy}
In this section, we explore the relation between our problem and that of separating colors with the minimum number of lines.
In particular, we show that greedily chosen separating lines induce an $O(\OptDJ\log^2{n})$-approximation to the minimum-cardinality island partition.

A set of lines $L$ \textit{separates} a set of colored points $S$ if each face in the arrangement $\mathcal{A}(L)$ is monochromatic.
The problem of finding the minimum-cardinality set of such separating lines is W[1]-hard with the parameter being the solution size~\cite{separate-by-lines-FPT}.
Furthermore, the problem is NP-hard~\cite{separate-by-lines-NP} and APX-hard~\cite{separate-by-lines-APX}, even when allowing only axis-parallel lines.
The problem can be viewed as a set cover problem where lines are used to cover line segments between pairs of points of different color.
Thus, the corresponding greedy algorithm, which we refer to as \textit{line-greedy}, yields a $O(\log{n})$-approximation~\cite{separate-by-lines-greedy1, separate-by-lines-greedy2}.
Line-greedy can be implemented to run in $O(k \OptL n^2 \log{n})$ time~\cite{separate-by-lines-greedy1}, where $\OptL$ is the optimal number of lines and $k$ is the number of colors of the input set.

If $L$ separates $S$, then the faces of the arrangement $\mathcal{A}(L)$ induce a partition of $S$ into $O(|L|^2)$ islands.
Conversely, an island partition $\mathcal{P}$ of $S$, with $|\mathcal{P}| \geq 3$, induces a set of $O(|\mathcal{P}|)$ lines that separates $S$.
This can be shown using a construction by Edelsbrunner, Robison, and Shen~\cite{cover-convex-sets-non-overlapping-polygons}.
We sketch their construction, adapted slightly for our use; see their paper for details and proofs. Circumscribe a rectangle around all the polygons---the convex hulls of the islands in $\mathcal{P}$.
Grow the polygons, by moving their sides, until they are maximal.
Extend each shared polygon side to obtain a set of lines.
This set of lines separates the input points.
Furthermore, each line corresponds to an edge of the contact graph of the expanded polygons (Figure~\ref{fig:grow-clusters}).
Because the contact graph is planar, there are at most $3|\mathcal{P}| - 6$ lines, yielding the desired result.
While the exact running time of the construction is unclear, it is clearly polynomial.
Pocchiola and Vegter~\cite{on-polygonal-covers} provide an alternative construction that makes use of a pseudo-triangulation of the polygons. Their algorithm runs in $O(n + |\mathcal{P}|\log{n})$ time.

\begin{figure}[tb]
    \hspace*{\fill}
    \includegraphics[page=1]{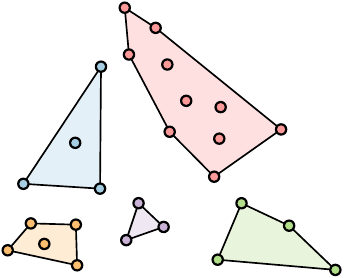}
    \hfill
    \includegraphics[page=2]{grow-clusters.pdf}
    \hspace*{\fill}
    \caption{Left: islands; right: expanded islands and their contact graph.}
    \label{fig:grow-clusters}
\end{figure}

Thus, an optimal island partition induces an $O(\OptL)$-approximation to the optimal set of separating lines.
Conversely, an optimal set of separating lines induces an $O(\OptDJ)$-approximation to the optimal island partition.
There is an analogous relation between approximation algorithms of the two problems. 
In particular, we have the following result.

\begin{restatable}{lemma}{lineGreedy}
Line-greedy induces an $O(\OptDJ \log^2{n})$-approximation to the minimum-cardinality island partition.
\end{restatable}
\begin{proof}
Line-greedy returns a set of lines $L$ of cardinality $O(\OptL \log{n})$. Creating an island for each face of $\mathcal{A}(L)$ that contains input points yields $O(\OptL^2 \log^2{n})$ islands. We have $\OptL \leq 3\OptDJ - 6$ because the optimal solution corresponding to $\OptDJ$ induces a set of $3\OptDJ - 6$ separating lines. Thus, this approach yields an $O(\OptDJ \log^2{n})$-approximation.
\end{proof}

For a lower bound instance, place points in a square grid of $k = 2^\ell$ rows and columns and color them alternatingly as in a checkerboard.
In addition, place points on the corners of thin axis-parallel rectangles on the sides of the grid to encourage the line-greedy algorithm to use axis-parallel lines (Figure~\ref{fig:line-greedy-lower}).
We suspect that for any $\ell \geq 1$ line-greedy returns horizontal and vertical separating lines that separate the rows and columns of the grid as shown on the left in Figure~\ref{fig:line-greedy-lower}.
However, a formal proof eludes us.
If this were true, then the island partition induced by the line-greedy solution would have cardinality $\Omega(k^2)$.
Because an island partition of cardinality $O(k)$ exists (Figure~\ref{fig:line-greedy-lower}, right), this would result in an $\Omega(\sqrt{n}) = \Omega(\OptDJ)$ lower bound on the approximation that is attained by an island partition induced by line-greedy.

\begin{figure}[tb]
    \hspace*{\fill}
    \includegraphics[page=5]{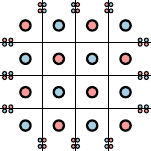}
    \hfill
    \includegraphics[page=6]{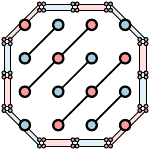}
    \hspace*{\fill}
    \caption{The figure shows an idea of a lower bound on the approximation that is attained by an island partition induced by line-greedy.}
    \label{fig:line-greedy-lower}
\end{figure}

\section{Conclusion}
As far as we are aware, this is the first paper that addresses the algorithmic problem of creating a minimum-cardinality partition of colored points into monochromatic islands.
We have proven bounds on the approximation ratios of three greedy algorithms: \emph{disjoint-greedy}, \emph{overlap-greedy}, and \emph{line-greedy}.
Of the three algorithms, our upper bound of $O(\OptDJ \log^2{n})$ on the approximation ratio of line-greedy is best.
To put this bound into perspective, note that the best bound on the approximation ratio of a greedy algorithm for the related, but arguably simpler, problem of covering colored points by the minimum number of pairwise-disjoint triangles is $O(\OptDJ \log^2{n})$ as well~\cite{geometric-partitions}. 
Agarwal and Suri~\cite{geometric-partitions} also proved an $O(\log{n})$-approximation bound for an $O(n^8)$ dynamic programming algorithm.
However, that approach does not seem applicable to our problem; in particular, the fact that islands are not of constant complexity complicates matters.
Avenues for future work are finding algorithms with better approximation ratios and settling the computational complexity of the problem.

\bibliography{bibliography}

\end{document}